\newif\ifdraft \draftfalse
\newif\iffull \fulltrue
\makeatletter \@input{tex.flags} \makeatother

\documentclass[11pt]{article}

\usepackage[affil-it]{authblk}
\usepackage[noend]{algorithmic}

\newcommand{\D}{\ensuremath{\mathcal{D}}\xspace}
\newcommand{\Dk}{\ensuremath{\mathcal{D}^k}\xspace}

\iffull
\usepackage{natbib}
\else
\usepackage[numbers,sort&compress]{natbib}
\fi

\usepackage{fullpage}
\usepackage[linesnumbered, ruled,vlined]{algorithm2e}
\usepackage{algorithmic}
\usepackage{amsmath, amssymb, amsthm}
\usepackage{mathtools}
\usepackage{verbatim}
\usepackage{color}
\usepackage{xcolor, color-edits}
\addauthor{ar}{purple}
\addauthor{sk}{red}
\addauthor{sw}{cyan}
\addauthor{jm}{magenta}
\usepackage{multicol}

\definecolor{DarkGreen}{rgb}{0.1,0.5,0.1}
\definecolor{DarkRed}{rgb}{0.5,0.1,0.1}
\definecolor{DarkBlue}{rgb}{0.1,0.1,0.5}
\usepackage[]{hyperref}
\hypersetup{
    unicode=false,          
    pdftoolbar=true,        
    pdfmenubar=true,        
    pdffitwindow=false,      
    pdftitle={},    
    pdfauthor={}
    pdfsubject={},   
    pdfnewwindow=true,      
    pdfkeywords={keywords}, 
    colorlinks=true,       
    linkcolor=DarkRed,          
    citecolor=DarkGreen,        
    filecolor=DarkRed,      
    urlcolor=DarkBlue,          
}
\usepackage{xspace}
\usepackage{cleveref}
\usepackage{thmtools, thm-restate}

\iffull
\usepackage{natbib}
\else
\usepackage[numbers,sort&compress]{natbib}
\fi

\newcommand\cA{\mathcal{A}}

\newcommand\cM{\mathcal{M}}
\newcommand\cP{\mathcal{P}}
\newcommand\cQ{\mathcal{Q}}
\newcommand\cR{\mathcal{R}}

\newcommand{\bin}{\text{Bin}}
\newcommand{\score}[2]{\texttt{score$(#1,#2)$}}
\newcommand{\schoolda}{\texttt{DA-School}\xspace}
\newcommand{\privateda}[2]{\texttt{Private-DA-School$(#1,#2)$}\xspace}
\newcommand{\counter}{{{\bf Counter}}\xspace}
\newcommand{\M}{\ensuremath{\mathcal{M}}\xspace}

\newcommand{\nomatch}{\odot}

\DeclareMathOperator{\poly}{poly}
\DeclareMathOperator{\polylog}{polylog}

\newcommand{\eps}{\varepsilon}

\def\epsilon{\varepsilon}

\DeclareMathOperator{\Lap}{Lap}

\renewcommand{\hat}{\widehat}

\DeclareMathOperator*{\argmax}{\mathrm{argmax}}
\DeclareMathOperator*{\secondargmax}{\mathrm{argsecondmax}}

\newcommand{\INDSTATE}[1][1]{\STATE\hspace{#1\algorithmicindent}}

\newtheorem*{theorem*}{Theorem}

\declaretheorem[
  name=Theorem,
  refname={theorem, theorems},
  Refname={Theorem, Theorems}]{theorem}
\declaretheorem[
  name=Lemma,
  refname={lemma, lemmas},
  Refname={Lemma, Lemmas}]{lemma}
\declaretheorem[
  name=Example,
  refname={example, examples},
  Refname={Example, Examples}]{example}

\declaretheorem[
  name=Remark,
  refname={remark, remarks},
  Refname={Remark, Remarks}]{remark}

\declaretheorem[
  name=Definition,
  refname={definition, definitions},
  Refname={Definition, Definitions}]{definition}

\title{Approximately Stable, School Optimal, and Student-Truthful Many-to-One
  Matchings\\ (via Differential Privacy)\thanks{Kannan was partially supported by NSF grant NRI-1317788. email: kannan@cis.upenn.edu.
    Morgenstern was partially supported by NSF
    grants CCF-1116892 and CCF-1101215, as well as a Simons Award for
    Graduate Students in Theoretical Computer Science.  Contact
    information: J. Morgenstern, Computer Science Department, Carnegie
    Mellon University, \texttt{jamiemmt@cs.cmu.edu}. Roth was partially supported by an NSF CAREER award, NSF Grants CCF-1101389 and
    CNS-1065060, and a Google Focused Research Award. Email: \texttt{aaroth@cis.upenn.edu}. Wu was supported in part by NSF Grants CCF-1101389. Email: \texttt{wuzhiwei@cis.upenn.edu}}}

\author[1]{Sampath Kannan}
\author[2]{Jamie Morgenstern}
\author[1]{Aaron Roth}
\author[1]{Zhiwei Steven Wu}

\affil[1]{Department of Computer and Information Science\\
 University of Pennsylvania}
\affil[2]{Computer Science Department\\
Carnegie Mellon University
}

\begin{document}
\maketitle

\begin{abstract}
  We present a mechanism for computing asymptotically stable school
  optimal matchings, while guaranteeing that it is an asymptotic
  dominant strategy for every student to report their true preferences
  to the mechanism. Our main tool in this endeavor is
  \emph{differential privacy}: we give an algorithm that coordinates a
  stable matching using differentially private signals, which lead to
  our truthfulness guarantee. This is the first setting in which it is
  known how to achieve nontrivial truthfulness guarantees for students
  when computing school optimal matchings, assuming \emph{worst-case}
  preferences (for schools and students) in large
  markets. 

\end{abstract}
\vfil
\pagebreak 

\section{Introduction}

In this paper we consider the important problem of computing
many-to-one stable matchings -- a matching solution concept used for
diverse applications, including matching students to schools
\citep{abdulkadiroglu2009strategy}, and medical residents to hospitals
\citep{roth1984evolution,roth2002redesign}. There are two sides of
such a market: we will refer to the side with multiple positions as
the \emph{schools} and the other side as the \emph{students}. The goal
is to find a feasible assignment $\mu$ of students to schools -- each
student $a$ can be matched to at most $1$ school, but each school $u$
can be potentially matched to up to $C_u$ students, where $C_u$ is the
\emph{capacity} of school $u$. We would like to find a matching that
is \emph{stable}. Informally, when each student $a$ has a preference
ordering $\succ_a$ over schools, and each school $u$ has a preference
ordering $\succ_u$ over students, then an assignment $\mu$ forms a
stable matching if it is feasible, and there is no student-school pair
$(a, u)$ such that they are unmatched $(\mu(a) \neq u)$, but such that
they would mutually prefer to deviate from the proposed matching $\mu$
and match with each other.

The set of stable many-to-one-matchings have a remarkable structural
property: there exists a \emph{school optimal} and a \emph{student
  optimal} stable matching -- i.e.\ a stable matching that all schools
simultaneously prefer to all other stable matchings, and a stable
matching that all students simultaneously prefer to all other stable
matchings. Moreover, these matchings are easy to find, with the
school-proposing (respectively, student proposing) version of the
Gale-Shapley deferred acceptance algorithm
\citep{gale1962college}. Unfortunately, the situation is not quite as
nice when student incentives are taken into account. Even in the
1-to-1 matching case (i.e.\ when capacities $C_u = 1$ for all schools),
there is no mechanism which makes truthful reporting of one's
preferences a dominant strategy for both sides of the market
\citep{roth1982economics}. In the many-to-one matchings case, things
are even worse: an algorithm which finds the school optimal stable
matching does not incentivize truthful reporting for either the
students or the schools~\citep{roth1984evolution}.

Because of this, a literature has emerged studying the incentive
properties of stable matching algorithms under \emph{large market
  assumptions}
(e.g.\ ~\citet{immorlica2005marriage,lee2011incentive,kojima2009incentives}). In
general, this literature has taken the following approach: make
restrictive assumptions about the market (e.g.\ that students
preference lists are only of constant length and are drawn uniformly
at random), and under those assumptions, prove that an algorithm which
computes exactly the school optimal stable matching makes truthful
reporting a dominant strategy for a $1 - o(1)$ fraction of student
participants (generally even under these assumptions, the schools
still have incentive to misreport if the algorithm computes the
school optimal stable matching).

In this paper we take a fundamentally different approach. We make
absolutely no assumptions about student or school preferences,
allowing them to be worst-case. We also insist on giving incentive
guarantees to \emph{every} student, not just most students. We compute
(in a sense to be defined) an \emph{approximately} stable and
\emph{approximately} school optimal matching using an algorithm with a
particular insensitivity property (differential privacy), and show
that truthful reporting is an \emph{approximately} dominant strategy
for \emph{every} student in the market. These approximations become
perfect as the size of the market grows large. Our notion of a ``large
market'' requires only that the capacity of each school $C_u$ grows with
(the square root of) the number of
schools, and (logarithmically) with the number of students, and does
not require any assumption on how preferences are generated.

\subsection{Our Results and Techniques}
We recall the standard notion of stability in a many-to-one matching
market with $n$ students $a_i \in A$ and $m$ schools $u_j \in U$, each
with capacity $C_j$.
\begin{definition} A matching $\mu:A\rightarrow U\cup \{\nomatch\}$ is
  \emph{feasible} and \emph{stable} if:
\begin{enumerate}
\item (Feasibility) For each $u_j \in U$, $|\{i : \mu(a_i) = u_j\}|
\leq C_j$ \label{cond:feasibility}
\item (No Blocking Pairs with Filled Seats) For each $a_i \in A$, and
each $u_j \in U$ such that $\mu(a_i) \neq u_j$, either $\mu(a_i)
\succ_{a_i} u_j$ or for every student $a'_i \in \mu^{-1}(u_j)$, $a'_i
\succ_{u_j} a_i$. \label{cond:filledblocking}
\item (No Blocking Pairs with Empty Seats) For every $u_j \in U$ such
that $|\mu^{-1}(u_j)| < C_j$, and for every student $a_i \in A$ such
that $a_i \succ_{u_j} \nomatch$, $\mu(a_i) \succ_{a_i}
u_j$. \label{cond:emptyblocking}
\end{enumerate}
\end{definition}

Our notion of approximate stability relaxes condition
\ref{cond:emptyblocking}. Informally, we still require that there be
no blocking pairs among students and \emph{filled} seats, but we allow
each school to possibly have a small number of empty seats. We view
this as a mild condition, reflecting the reality that schools are not
able to perfectly manage yield, and are often willing to accept a
small degree of under-enrollment.

\begin{definition}[Approximate Stability]
\label{def:approxstable} A matching $\mu:A\rightarrow U\cup \{\nomatch\}$ is
\emph{feasible} and $\alpha$-approximately \emph{stable} if it
satisfies conditions \ref{cond:feasibility} and
\ref{cond:filledblocking} (\emph{Feasibility} and \emph{No Blocking
Pairs with Filled Seats}) and:
\begin{enumerate} \setcounter{enumi}{2}
  \item (No Blocking pairs with Empty Seats at Under-Enrolled
Schools) For every $u_j \in U$ such that $|\mu^{-1}(u_j)| <
(1-\alpha)C_j$, and for every student $a_i \in A$ such that $a_i
\succ_{u_j} \nomatch$, $\mu(a_i) \succ_{a_i}
u_j$. \label{cond:approxemptyblocking}
\end{enumerate}
\end{definition}

We also employ a strong notion of approximate dominant strategy
truthfulness, related to first order stochastic dominance --
informally, we say that a mechanism is $\eta$-approximately dominant
strategy truthful if no agent can gain more than $\eta$ in expectation
(measured by \emph{any} cardinal utility function consistent with his
ordinal preferences) by misreporting his preferences to the mechanism.

Finally, we, define a notion of school optimality that applies to
approximately stable matchings. Informally, we say that an
approximately stable matching $\mu$ (in the above sense) is
\emph{school dominant} if when compared to the school optimal
\emph{exactly} stable matching $\mu'$, for every school $u_j$, every
student $a_i$ matched to $u_j$ in $\mu$ is strictly preferred by $u_j$
to any student matched to $\mu_j$ in $\mu'$ but not in $\mu$.

We can now give an informal statement of our main result.

\begin{theorem}[Informal]\label{thm:school-truth}
  There is an algorithm for computing feasible and
  $\alpha$-approximately stable \emph{school dominant} matchings that
  makes truthful reporting an $\eta$-approximate dominant strategy for
  every student in the market, under the condition that for every
  school $u$, the capacity is sufficiently large:
$$C_u = \Omega\left(\frac{\sqrt{m}}{\eta \alpha}\cdot \mathrm{polylog}(n)\right)$$
\end{theorem}

\begin{remark}
  Note that no assumptions are needed about either school or student
  preferences, which can be arbitrary. The only large market
  assumption needed is that the capacity $C_u$ of each school is
  large. If, as the market grows, school capacities grow slightly
  faster than the square root of the number of schools, then both $\eta$
  and $\alpha$ can be taken to tend to 0 in the limit.
\end{remark}

This result differs from the standard large market results in several
ways. First, and perhaps most importantly, the result is worst-case
over all possible preferences of both schools and students. Second,
the guarantee states that \emph{no} student may substantially gain by
by misreporting her preferences; previous
results~\citep{immorlica2005marriage, kojima2009incentives} show that
only a subconstant fraction of students might have (substantial)
incentive to deviate. In exchange for these strong guarantees, we
relax our notion of stability and school optimality to approximate
notions, which can be taken to be exact in the limit as the market
grows large (under the condition that school capacities grow at a
sufficiently fast rate).

When we do make some of the assumptions on student preferences made in
previous work, we get stronger claims than the one above. For example,
when the length of the preference lists of students are taken to be
bounded, as they are in~\cite{immorlica2005marriage}
and~\cite{kojima2009incentives} we can remove our dependence on the
number of schools:

\begin{theorem}[Informal]\label{thm:k-list}
  Under the condition that all students have preference lists over at
  most $k$ schools (and otherwise prefer to be unmatched), there is an
  algorithm for computing feasible and $\alpha$-approximately stable
  \emph{school dominant} matchings that makes truthful reporting an
  $\eta$-approximate dominant strategy for every student in the
  market, under the condition that for every school $u$, the capacity
  is sufficiently large:
$$C_u = \Omega\left(\frac{\sqrt{k}}{\eta \alpha}\cdot \mathrm{polylog}(n)\right)$$
\end{theorem}
\begin{remark}
  Note that if $k$ is considered to be a constant, then this result
  requires school capacity to grow only poly-logarithmically with the
  number of students $n$.
\end{remark}

Our results come from analyzing a differentially private variant of
the classic deferred acceptance algorithm. Rather than having schools
explicitly \emph{propose} to students, we consider an equivalent
variant in which schools $u$ publish a set of ``admissions
thresholds'' which allow any student $a$ who is ranked higher than the
current threshold of school $u$ (according to the preferences of $u$)
to enroll. These thresholds naturally induce a matching when each
student enrolls at their favorite school, given the thresholds. We
first show that if the thresholds are computed under the constraint of
differential privacy, then the algorithm is approximately dominant
strategy truthful for the students. We then complete the picture by
deriving a differentially private algorithm, and showing that with
high probability, it produces an approximately stable, school dominant
matching.

\section{Related Work}

\subsection{Incentives in Stable Matching}
Stable matching has long been known to be incompatible with
truthfulness: no algorithm which produces a stable matching is
truthful for both sides of the market~\citep{roth1982economics},
though Gale-Shapley is known to be truthful for the side of the
market which is proposing.  Several lines of work have investigated
stable matching in \emph{large} markets, where players' preferences
are drawn from some distribution, and considering properties of the
market as $n$, the number of players, grows large.

Let \D be a fixed distribution over the set of $n$ women. Consider the
following process of generating length-$k$ preference lists over women.
Draw some $w_1 \sim \D$, and let $w_1$ be the first woman in a
preference list. Now, let $(w_1, \ldots, w_{i-1})$ be the first $(i
-1)$ women, in order, drawn from \D. Draw $w_i\sim \D$ until $w_i
\notin \{w_1, \ldots, w_{i-1}\}$. We denote such a distribution
over preference lists by \Dk. \citep{immorlica2005marriage} prove a
generalization of a conjecture of ~\citet{roth2002redesign}, showing
if the men draw their preference lists according to \Dk, the expected
number of women with more than one stable match is $o(n)$ (as $n$
grows, for fixed $k$). Since it is known that a person has incentive
to misreport only if they have more than one stable partner, this
implies that only a vanishingly small fraction of the women will have
incentive to misreport to any stable matching process. They also show
that any stable matching algorithm induces a Nash equilibrium for
which a $1-o(1)$ fraction of players behave truthfully.

\citet{kojima2009incentives} generalized these results to the
many-to-one matching setting.\footnote{Here, a school has two types of
  misreporting: first, it might misreport its capacity for students,
  and second, it might misreport its preference list over students.}
In particular, they consider student-optimal stable matchings, where
colleges have arbitrary preferences, and the students have random
preference lists of fixed length drawn as above. They show that the
number of schools which have incentive to misreport in the
student-optimal matching is $o(n)$ (again, as $n$ grows, for constant
$k$). The paper mentions that it is possible to consider the
school-optimal stable matching and define things analogously, but that
it is not clear that one can derive truthfulness guarantees for the
schools in this setting without additional assumptions. Under further
assumptions about the ``thickness" of popular schools of the
distribution, they show that truthtelling forms an
approximate Bayes Nash equilibrium as the number of
colleges grows.

\begin{table*}[t]
\begin{center}
\begin{tabular}{ | l | p{4cm} | p{5cm} |}
  \hline
      Reference & Assumptions & \# of possible stable matches \\ \hline
      \citet{immorlica2005marriage} & Random, i.i.d. preferences on male side &
      $1-o(1)$-fraction of women have $\leq 1$ stable match \\\hline
      \citet{kojima2009incentives} & Random, i.i.d. preferences on student side &$o(1)$-fraction of schools have more than $1$ stable set of students\\ \hline
      \citet{pittel1992likely} &  Uniform random preferences on both sides & Average rank of parter for side optimized for is $\log(n)$, $\frac{n}{\log(n)}$ for the non-optimized side\\ \hline
      \citet{ashlagi2013unbalanced} & Uniform random preferences on both sides, $n$ men, $n-1$ women & Average rank for men's match $\frac{n}{3\log(n)}$, for women's match $3\log(n)$, in any stable matching \\ \hline
    \end{tabular}
    \caption{Related works with distributional assumptions on the
      preferences of one or both sides of the market. Under these
      assumptions, it is often possible to show that many agents have
      few (or one) stable partners. If an agent has zero or one stable
      partner, then she has no incentive to
      misreport.} \label{table:numbers}
\end{center}
\end{table*}

\citet{lee2011incentive} considers a different method for constructing
preference orderings in the one-to-one matching setting. Each man $m$
has an intrinsic value $V_m$, such that for each woman $w$, her
utility from being matched to $m$ is $U(V_m, \eta_{m,w})$, where
$\eta_{m,w}$ is a draw from $w$'s private distribution for $m$. This
assumption is made for both sides of the market's preferences. Under
these assumptions, only a small number of players have incentive to
misreport, and there is an $\epsilon$-Nash equilibrium where almost
all players report truthfully.

\citet{azevedo2012strategyproofness} introduce the notion of
``strategyproofness in the large'', and show that the Gale-Shapley
algorithm satisfies this definition. Roughly, this means that fixing
any (constant sized) typespace, and any distribution over that
typespace, if player preferences are sampled i.i.d. from the
typespace, then for any fixed $\eta$, as the number of players $n$
tends to infinity, truthful reporting becomes an $\eta$-approximate
Bayes Nash equilibrium. These assumptions can be restrictive however
-- note that this kind of result requires that there are many more
players than there are ``types'' of preferences, which in particular
(together with the full support assumption on the type distribution)
requires that in the limit, there are infinitely many identical agents
of each type. In contrast, our results do not require a condition like
this.

To the best of our knowledge, our results are the first to give
truthfulness guarantees in settings where both sides of the market
have worst-case preference orderings. Unlike some prior work, even
without distributional assumptions, we are able to give truthfulness
guarantees to \emph{every} student, not only a $1-o(1)$ fraction of
students. Under some of the distributional assumptions used in prior
work, our results can be sharpened as well.

\subsection{Differential Privacy as a Tool for Truthfulness}

The study of differentially private
algorithms~\cite{dwork2006calibrating} has blossomed in recent years.
A comprehensive survey of the work in this area is beyond the scope of
this paper; here, we mention the work which relates directly to the
use of differential privacy in constructing truthful mechanisms.

~\citet{mcsherry2007truthful} were the first to identify privacy as a
tool for designing approximately truthful
mechanisms. \citet{nissim2012privacy} showed how privacy could be used
as a tool to design \emph{exactly} truthful mechanisms without needing
monetary payments (in certain settings). \citet{huang2012exponential}
proved that the exponential mechanism, a basic tool in differential
privacy introduced in~\citet{mcsherry2007truthful} is maximal in
distributional range, which implies that there exist payments which
make it exactly truthful. \citet{kearns2014large} demonstrated a
connection between private equilibrium computation and the design of
truthful \emph{mediators} (and also showed how to privately compute
approximate correlated equilibria in large games). This work was
extended by~\citet{rogers2014congestion} who show how to privately
compute \emph{Nash} equilibria in large congestion games.

The paper most related to our own is~\citet{hsu2013walrasian} which
shows how to compute approximate Walrasian equilibria privately, when
bidders have quasi-linear utility for money and the supply of each
good is sufficiently large. In that paper, in the final allocation,
every agent is matched to their \emph{approximately} most preferred
goods at the final prices. In our setting, there are several
significant differences: first, in the Walrasian equilibrium setting,
only agents have preferences over goods (i.e. goods have no
preferences of their own), but in our setting, both sides of the
market have preferences. Second, although there is a conceptual
relationship between ``threshold scores'' in stable matching problems
and prices in Walrasian equilibria, the thresholds do not play the
role of money in matching problems, and there is no notion of being
matched to an ``approximately'' most preferred school.

\section{Preliminaries}

\subsection{Many-to-one Matching}
A many-to-one stable matching problem consist of $m$ schools $U=\{u_1,
\ldots , u_m\}$ and $n$ students $A=\{a_1, \ldots , a_n\}$. Every
student $a$ has a preference ordering $\succ_a$ over all the schools,
and each school $u$ has a preference ordering $\succ_u$ over the
students. Let $\cP$ denote the domain of all preference orderings over
schools (so each $\succ_a \in \cP$). It will be useful for us to think
of a school $u$'s ordering over students $A$ as assigning a
unique\footnote{It is essentially without loss of generality that
  students are assigned unique scores. If not, we could break ties by
  a simple pre-processing step: add noise $\sum_{k=1}^l 2^{-k}b_k$ to
  each student's score, where each $b_k$ is a random bit; if the
  scores are integral, the probability of having ties is $1/\poly(n)$
  as long as $l\geq O(\log(n))$. } \emph{score} $\score{u}{a}$ to
every student, in descending order (for example, these could be
student scores on an entrance exam). Let the range of these scores be
$\score{u}{a} \in \cQ$ and $\cQ^n$ be the set of all valid score
profiles for a given school. We will overload the notation of $\cQ^n$
to refer both to the set of admissible score vectors and the set of
preference orderings over students when our use is clear from context.

Every school $u$ has a capacity $C_u$, the maximum number of students
the school can accommodate. A feasible matching $\mu$ is a mapping
$\mu:A\rightarrow U\cup\nomatch$, which has the property each student
$a$ is paired with at most one school $\mu(a)$, and each school $u$ is
matched with at most $C_u$ students: $|\mu^{-1}(u)| \leq C_u$. If
$\mu(a) = \nomatch$, we say $a$ is \emph{unmatched} by $\mu$. For
notational simplicity, we will sometimes simply write $\mu(u)$ to
denote the set of students assigned to school $u$.

A matching is $\alpha$-approximately \emph{stable} if it satisfies
 Definition~\ref{def:approxstable}. When computing matchings, it will
be helpful for us to think instead about computing \emph{admission
  thresholds} $t_{u}$ for each school. A set of admission thresholds $t
\in \mathbb{R}_{\geq 0}^m$ induces a matching $\mu$ in a natural way:
every student $a \in A$ is matched to her most preferred school
amongst those whose admissions thresholds are below her score at the
school. Formally, for a set of admissions thresholds $t$, the induced
matching $\mu^t$ is defined by:
$$\mu^t(a) =  \arg\max_{\succ_a}\{u \mid  \score{u}{a} \geq t_u\}$$

We say that a set of admission thresholds $s$ is feasible and
$\alpha$-approximately stable if its induced matching $\mu^t(a)$ is
feasible and $\alpha$-approximately stable.

Note that an $\alpha$-stable matching is an exactly stable matching in
a market in which schools have reduced capacity (where the capacity at
each school is reduced by at most a $(1-\alpha)$ factor).

\paragraph{Remark}{Definition~\ref{def:approxstable} also implies that
  if a school $u$ is under-enrolled by more than $\alpha C_u$, its
  admission score $t_u = 0$. This means such a school is very
  unpopular and could not recruit enough students even without any
  admission criterion. }

We now introduce a notion of approximate school optimality, which our
algorithms will guarantee.

\begin{definition}\label{def:school-dom}
  A matching $\mu$ is {\bf school-dominant} if, for each school $u$,
  for all $a\in \mu(u) \setminus \mu'(u)$ and all $a'\in
  \mu'(u)\setminus \mu(u)$, $a \succ_u a'$, where $\mu'$ is the
  school-optimal matching.
\end{definition}

In words, a matching $\mu$ is school-dominant if for every school $u$,
when comparing the set of students $S_1$ that $u$ is matched to in $\mu$
but not in the school optimal matching $\mu'$, and the set of students
$S_2$ that $u$ is matched to in the school optimal matching $\mu'$, but
is not matched to in $\mu$, $u$ strictly prefers every student in $S_1$
to every student in $S_2$. (i.e. compared to the school optimal
matching, a school may be matched to \emph{fewer} students, but not to
\emph{worse} students.) We note that school-dominance alone is trivial
to guarantee: in particular, the empty matching is school dominant.
Only together with an upper bound on the number of empty seats allowed
per school (for example, as guaranteed by $\alpha$-approximate
stability) is this a meaningful concept.

We want to give mechanisms that make it an \emph{approximately
  dominant strategy} for students to report truthfully. We have to be
careful about what we mean by this, since students $a$ have ordinal
preferences $\succ_a$, rather than cardinal utility functions
$v_a:U\rightarrow [0,1]$. We say that a cardinal utility function
$v_a$ is \emph{consistent} with a preference ordering $\succ_a$ if for
every $u, u' \in U$, $u \succ_a u'$ if and only if $v_a(u) \geq
v_a(u')$. We will say that a mechanism is $\eta$-approximately
truthful for students if for every student, and \emph{every cardinal
  utility function $v_a$ consistent with truthful $\succ_a$}, truthful
reporting is an $\eta$-approximate dominant strategy as measured by
$v_a$.

\begin{definition}
  Consider any randomized mapping $\mathcal{M}:\mathcal{P}^n\times
  {\cQ^n}^m\rightarrow (U\cup \nomatch)^n$. We say that $\mathcal{M}$
  is $\eta$-approximately dominant strategy truthful for students (or
  student-truthful) if for any vector of school and student
  preferences $\succ \in (\mathcal{P}^n\times {\cQ^n}^m)$, any student
  $a$, any utility function $v_a:U\rightarrow [0,1]$ that is
  consistent with $\succ_a$, and any $\succ_a' \neq \succ_a$, we have:
$$\mathrm{E}_{\mu \sim M(\succ)}[v_a(\mu(a))] \geq \mathrm{E}_{\mu
  \sim  M(\succ_a', \succ_{-a})}[v_a(\mu(a))] - \eta.$$
\end{definition}

Note that this definition is very strong, since it holds
simultaneously for every utility function consistent with student
preferences. When $\eta = 0$ it corresponds to first order stochastic
dominance.

\subsection{Differential Privacy Preliminaries}
Our tool for obtaining approximate truthfulness is \emph{differential
  privacy}, which we define in this section. We say that the ``private
data'' of each student $a$ consists of both her preference ordering
$\succ_a\in \cP$ over the schools and her scores $\score{u}{a}\in
\cQ$,\footnote{This means a student's set of scores, one from each
  school, can be written as $\cQ^m$.} one assigned by each school. A
private database $D \in (\cP\times \cQ^m)^n$ is a vector of $n$
students' preferences and scores, and $D$ and $D'$ are neighboring
databases if they differ in no more than one student record. In
particular, our matching algorithm takes as input $n$ students'
private data, and outputs a set of threshold admission scores (i.e.,
the mechanism's range is $\cR = \cQ^m$).

\begin{definition}[\citet{dwork2006calibrating}]
  An (randomized) algorithm $\cA \colon (\cP \times \cQ^m)^n\rightarrow
  \cR$ is {\em $(\epsilon,\delta)$-differentially private} if for
  every pair of neighboring databases $D, D' \in (\cP \times \cQ^m)^n$
  and for every subset of outputs $S \subseteq \cR$,
\[
  \Pr[\cM(D) \in S] \leq \exp(\epsilon)\Pr[\cM(D') \in S] + \delta.
\]
  If $\delta = 0$, we say that $\cM$ is {\em $\epsilon$-differentially private}.
\end{definition}

\subsection{Differentially Private Counters}

The central privacy tool in our matching algorithm is the private
streaming counter\footnote{ For a more detailed discussion of
  differential privacy under continual observation, see
  \citet{chan-counter} and \citet{DNPR10}.} proposed by
\citet{chan-counter} and \citet{DNPR10}. Given a bit stream $\sigma =
(\sigma_1, \ldots , \sigma_T)\in \{-1,0,1\}^T$, a streaming counter
$\cM(\sigma)$ releases an approximation to $c_\sigma(t) =
\sum_{i=1}^t\sigma_i$ at every time step $t$. Below, we define an
accuracy property we will then use to describe the usefulness of these
counters.

\begin{definition}
  A streaming counter $\cM$ is {\em $(\tau, \beta)$-useful} if with
  probability at least $1 - \beta$, for each time $t \in [T]$,
  \[
    \left| \cM(\sigma)(t) - c_\sigma(t) \right| \leq \tau.
  \]
\end{definition}

For the rest of this paper, let $\counter(\epsilon, T)$ denote the
Binary Mechanism of \citet{chan-counter}, instantiated with parameters
$\epsilon$ and $T$. $\counter(\epsilon, T)$ satisfies the following
accuracy guarantee (further details may be found in
Appendix~\ref{counter-details}).

\begin{restatable}[\citet{chan-counter}]{theorem}{counteraccuracy}
  \label{counter-error}
  For $\beta > 0$, $\counter(\epsilon, T)$ is
  $\epsilon$-differentially private with respect to a single bit change in the
  stream, and $(\tau, \beta)$-useful for
  \[
    \tau = \frac{4\sqrt{2}}{\epsilon} \ln \left( \frac{2}{\beta}\right)
    \left(\sqrt{\log(T)}\right)^5.
   \]
\end{restatable}

Our mechanism uses $m$ different $\counter$s to maintain the counts of
tentatively enrolled students for all schools. The following theorem
allows us to bound the error of each \counter through the collective
sensitivity across all $\counter$s.
\begin{restatable}{theorem}{countercomp}
\label{thm:counters}
Suppose we have $m$ bit streams such that the change of an student's
data affects at most $k$ streams, and alters at most $c$ bits in each
stream. For any $\beta > 0$, the composition of $m$ distinct
  $\counter\left(\eps/2c\sqrt{2kc\ln(1/\delta)}, T\right)$'s is $(\eps,
  \delta)$-differentially private, and $(\tau, \beta)$-useful for
  \[
  \tau =
  \frac{16c\sqrt{kc\ln(1/\delta)}}{\eps}\ln\left(\frac{2}{\beta}
  \right)\left( \sqrt{\log(T)} \right)^5.
  \]
\end{restatable}

%

\section{Algorithms Computing Private Matchings are Approximately Truthful}
In this section, we prove the theorem which motivates the rest of our
paper. Consider an algorithm $M$ which takes as input student and
school preferences $\succ$ and computes school thresholds $t$. If $M$
is $(\epsilon, \delta)$-differentially private, then the algorithm which
computes thresholds $t = M(\succ)$ and then outputs the induced
matching $\mu^t$ is $(\epsilon + \delta)$-approximately dominant strategy
truthful. Note that this guarantee holds independent of stability.

\begin{theorem}
  Let $M:(\mathcal{P}\times \cQ^m)^n\rightarrow \mathbb{R}_{\geq 0}^m$
  be any $(\epsilon, \delta)$-differentially private mechanism which
  takes as input $n$ student profiles and outputs $m$ school
  thresholds. Let $F_{\succ}:\mathbb{R}_{\geq 0}^m \rightarrow (U\cup
  \nomatch)^n$ be the function which takes as input $m$ school
  thresholds $t$ and outputs the corresponding matching $F_{\succ}(t)
  = \mu^t$. Then the mechanism $F_\succ\circ
  M:(\mathcal{P}\times\cQ^m)^n\rightarrow U^n$ is $(\epsilon +
  \delta)$-approximately dominant strategy truthful.
\end{theorem}
\begin{proof}
  Fix any vector of preferences $\succ$, any student $a$, any utility
  function $v_a$ consistent with $\succ_a$, and any deviation
  $\succ_a' \neq \succ_a$. Let $\succ' =(\succ'_a, \succ_{-a})$ for
  brevity. Now consider student $a$'s utility for truthtelling. For
  $\eps \leq 1$, we have

\begin{align*}
 & \mathrm{E}_{\mu \sim F_\succ\circ M(\succ)}[v_a(\mu(a))]\\
 = &\mathrm{E}_{t \sim M(\succ)} [v_a(\arg\max_{\succ_a}\{u \mid  \score{u}{a} \geq t_u\})] \\
  = & \sum_t \mathbb{P}[M(\succ) = t]\cdot v_a(\arg\max_{\succ_a}\{u \mid  \score{u}{a} \geq t_u\}) \\
  \geq & \sum_te^{-\epsilon} \mathbb{P}[M(\succ') = t] v_a(\arg\max_{\succ_a}\{u \mid  \score{u}{a} \geq t_u\}) - \delta\\
  =& e^{-\epsilon} \mathrm{E}_{t \sim M(\succ')}
  [v_a(\arg\max_{\succ_a}\{u \mid  \score{u}{a} \geq t_u\})]  - \delta\\
  \geq& e^{-\epsilon} \mathrm{E}_{t \sim M(\succ')}
  [v_a(\arg\max_{\succ_a'}\{u \mid  \score{u}{a} \geq t_u\})]  - \delta\\
  \geq& (1-\epsilon) \mathrm{E}_{t \sim M(\succ')}
  [v_a(\arg\max_{\succ_a'}\{u \mid  \score{u}{a} \geq t_u\})]  - \delta\\
  \geq& \mathrm{E}_{\mu \sim  M(\succ')}[v_a(\mu(a))] -
  (\eps + \delta),
\end{align*}
where the first and last equalities follow from the definition of
the induced matching $\mu^t$, the first inequality follows from the
differential privacy condition, and the second follows from the consistency of $v_a$ with $\succ_a$.
\end{proof} 
\section{Truthful School-Optimal Mechanism}

In this section, we present the algorithm which proves our main result
Theorem~\ref{thm:school-truth}. Theorem~\ref{alg:school-propose} computes an
$\alpha$-approximately stable and school-dominant matching, and enjoys
approximate dominant strategy truthfulness for the student side.

Before we do this, we present \schoolda, the well-known
school-proposing deferred acceptance
algorithm~\citep{gale1962college}. In this setting, schools which are
not at capacity \emph{propose} to students one at a time, starting
from their favorite students and moving down their preference
list. When a student gets a proposal, if she is tentatively matched to
some other school, she will reject the offer from whichever school she
likes less and accept the offer from the school she likes better. At
this point, she is tentatively matched to the school she likes better,
and the other school will continue to make proposals to fill the seat
offered to her. The version of the algorithm we present here is
non-standard -- it operates by having each school set an admissions
threshold, which it decreases slowly -- but is easily seen to be
equivalent to the deferred acceptance algorithm. This version of the
algorithm will be much more amenable to a private implementation,
which we give next. When a school $u$ lowers its threshold $t_u$ below
the score of a student $a$ at school $u$ ($\score{u}{a}$), we say that
school $u$ has \emph{proposed} to student $a$.

\begin{algorithm}
\caption{\schoolda, the deferred acceptance algorithm with schools proposing}\label{alg:da}
  \begin{algorithmic}
    \STATE{\textbf{Input: } school capacities $\{C_u\}$, student
      preferences $\{\succ_a\}$ and scores $\{\score{u}{a}\}$, \\range
      of scores $[0, J]$ } \STATE{\textbf{Output: } a set of score
      thresholds $\{t_j\}$} \STATE{\textbf{initialize: for each}
      school $u_j$ and student $a_i$} \STATE{ $\mbox{counter}(u_j) = 0
      \qquad t_j = J \qquad \mu(a_i) = \nomatch$ }
    \STATE{\textbf{while} there is some under-enrolled school $u_j$:
      counter$(u_j) \leq \hat{C}_{u_j}$ and $t_{u_j} > 0$}
    \INDSTATE{\textbf{do} $t_{u_j} = t_{u_j} - 1$}
    \INDSTATE{\textbf{for all} student $a_i$}
    \INDSTATE{\textbf{if} $\mu(a_i) \neq \arg\max_{\succ_{a_i}} \{u_j \mid
      \score{u_j}{a_i} \geq t_{u_j}\}$} 
    \INDSTATE[2]{\textbf{then} counter$(\mu(a_i))$
      = counter$(\mu(a_i)) - 1$}\;
    \INDSTATE[2]{$\mu(a_i) =
      \argmax_{\succ_{a_i}} \{u_j \mid \score{u_j}{a_i} \geq
      t_{u_j}\}$}\;
    \INDSTATE[2]{counter$(\mu(a_i))$ = counter$(\mu(a_i))   +1$}
    \RETURN{Final threshold scores $\{t_j\}$}
    \end{algorithmic}
  \end{algorithm}

  It is well-known that \schoolda will output a school-optimal stable
  matching (in our notation, a $0$-approximate school-dominant stable
  matching)~\citep{rothcollege}, assuming all players are truthful.
  We now state a useful fact about deferred acceptance, whose
  proof can be found in Appendix~\ref{sec:matching}.
%
%
%
%
%
%

\begin{restatable}{lemma}{prefix}
  \label{lem:prefix}
  Let $\mu_t$ be some matching which is an intermediate matching in a
  run of the school-proposing deferred acceptance algorithm. Then
  $\mu_t$ is school-dominant.
\end{restatable}
%
%
Our algorithm, \privateda{\eps}{\delta}, is a
private version of \schoolda. At each time $t$, each school will
publish a threshold score (initially, for each school, this will be
the maximum possible score for that school). Schools will lower their
thresholds when they are under capacity; as they do so, some students
will tentatively accept admission and some will reject or leave for
other schools. Initially, all students will be unmatched. For a given
student $a$, as soon as a school lowers its threshold below the score
$a$ has there, $a$ will signal to the mechanism which school is her
favorite of those for which her score passes their threshold. Then, as
the schools continue to lower their thresholds to fill seats, if a
school that $a$ likes better than her current match lowers its
threshold below her score, $a$ will inform the mechanism that she
wishes to switch to her new favorite.

Each school maintains a private counter of the number of students
tentatively matched to the school. We let $E$ be the additive error
bound of the counters. The schools will reserve $E$ number of
seats from their initial capacity to avoid being over-enrolled, so the
algorithm is run as if the capacity at each school is $C_u - E$. Then
each school can be potentially under-enrolled by $2E$ seats, but they
would take no more than $\alpha$ fraction of all the seats as long as
the capacity $C_u \geq 2E / \alpha$.

\begin{algorithm}[h!]
  \begin{algorithmic}
    \STATE{\textbf{Input: } school capacities $\{C_u\}$, student
      preferences $\{\succ_a\}$ and scores $\{\score{u}{a}\}$, \\range
      of scores $[0, J]$ }
    \STATE{\textbf{Output: } a set of score thresholds $\{t_j\}$}
    \STATE{\textbf{initialize:} }
    \STATE{{\small $T = m n J \qquad \eps' = \frac{\eps}{16\sqrt{2 m\ln(1/\delta)}}$}}
    \STATE{\small{$E =  \frac{128\sqrt{m\ln(1/\delta)}}{\eps} \ln\left(\frac{2m}{\beta}\right)\left( \sqrt{\log(nT)} \right)^5$}}
    \STATE{\textbf{for each} school $u_j$ and student $a_i$}
    \STATE{\centering$\mbox{counter}(u_j) = \counter(\eps', n T)\qquad
      t_j = J$}
    \STATE{\centering$\mu(a_i) = \nomatch\qquad \hat{C}_{u_j} = C_{u_j} - E$}

    \STATE{\textbf{while} there is some under-enrolled school $u_j$: counter$(u_j) <
      \hat{C}_{u_j}$ and $t_{u_j} > 0$}
     \INDSTATE{\textbf{do} $t_{u_j} = t_{u_j} - 1$}
     \INDSTATE{\textbf{for all} student $a_i$}
     \INDSTATE{\textbf{if }{\small$\mu(a_i) \neq \argmax_{\succ_{a_i}} \{u_j \mid  \score{u_j}{a_i} \geq t_{u_j}\}$}}
     \INDSTATE{\textbf{then }Send $(-1)$ to counter$(\mu(a_i))$}\;
     \INDSTATE{$\mu(a_i) = \argmax_{\succ_{a_i}} \{u_j \mid  \score{u_j}{a_i} \geq t_{u_j}\}$}\;
     \INDSTATE{Send $1$ to counter$(\mu(a_i))$}\;
     \INDSTATE{Send $0$ to all other counters}
     \INDSTATE{\textbf{else }Send $0$ to all counters}
    \RETURN{Final threshold scores $\{t_j\}$}
    \end{algorithmic}
\caption{\privateda{\eps}{\delta}, \label{alg:school-propose}}
  \end{algorithm}

 Now, we state the formal version of Theorem~\ref{thm:school-truth}.
\begin{theorem}\label{thm:formal}
 {
   \privateda{\eps}{\delta} is $(\eps, \delta)$-differentially private, and hence $(\eps +
   \delta)$-approximately dominant strategy truthful. With probability
   at least $1-\beta$, it outputs a set of $\alpha$-approximately
   stable admission thresholds that induces a school-dominant
   matching, as long as the capacity at each school $u$ satisfies
\[
C_u = \Omega\left( \frac{\sqrt{m}}{\eps \alpha} \polylog\left(n, m,
    \frac{1}{\delta}, \frac{1}{\beta} \right) \right).
\]
 }
\end{theorem}
%
%
%

We prove Theorem~\ref{thm:formal} in two parts.
Lemma~\ref{lem:private} shows that \privateda{\epsilon}{\delta} is
$(\epsilon,\delta)$-differentially private in the students' data.
Lemma~\ref{lem:school-opt} shows that the resulting matching is
school-dominant so long as the capacity at each school is large
enough. These two together imply Theorem~\ref{thm:formal}
directly.

\begin{lemma}\label{lem:private}
  \privateda{\eps}{\delta} is $(\epsilon, \delta)$-differentially
  private.
\end{lemma}

\begin{lemma}\label{lem:school-opt}
  With probability at least $1-\beta$, \privateda{\eps}{\delta}
  outputs a set of $\alpha$-approximately stable admission thresholds
  that induces a school-dominant matching, as long as the capacity at
  each school $u$ satisfies
\[
C_u = \Omega\left( \frac{\sqrt{m}}{\eps \alpha} \polylog\left(n, m,
    \frac{1}{\delta}, \frac{1}{\beta} \right) \right).
\]

\noindent If the maximum error of each of the collection of counters is bounded
by $x$ with probability at least, $1-\beta$, we need only 

\[
C_u = \Omega\left(x/ \alpha\right)
\]
\end{lemma}

\begin{proof}[Lemma~\ref{lem:private}]
  \privateda{\eps}{\delta} outputs a sequence of sets of thresholds
  and nothing else. We will construct a mechanism $\M$, which will
  output the same sequence of thresholds as \privateda{\eps}{\delta},
  for which it is more obvious to prove $(\epsilon,
  \delta)$-differential privacy. This will imply $(\epsilon,
  \delta)$-differential privacy of \privateda{\eps}{\delta}.
  Here is the definition of \M.
\begin{algorithm}\caption{\M}
\begin{algorithmic}
  \STATE{\textbf{initialize:} for each school $u_j$}
  \STATE{Threshold $t_{u_j} = J$}\STATE{Capacity$\hat{C}_{u_j} = C_{u_j} - E$}
  \STATE{$\eps' = \frac{\eps}{16\sqrt{2m\ln\frac{1}{\delta}}}$}
  \STATE{$\text{counter}(u_j) = \counter(\epsilon', nmJ)$}
  \STATE{\textbf{while} there is some under-enrolled school $u_j$:
  counter$(u_j) \leq \hat{C}_{u_j}$ and $t_{u_j} > 0$}
\INDSTATE[1]{Let $t'_{u_j} = t_{u_j} - 1$}\;
\INDSTATE[1]{Publish thresholds $(t_{u_1}, \ldots, t'_{u_j}, \ldots, t(u_m))$}\;
\INDSTATE[1]{Receive bits $b_{u_j'} \in \{-1,0,1\}$ for each $u_{j'}$}\;
\INDSTATE[1]{Send $b_{u_j}$ to counter($u_j$)}\;
\end{algorithmic}
\end{algorithm}

We define the input bits to the algorithm \M as follows. For a fixed
execution of the while loop, we will define the bits $b_{u_{j'}}$ to
give to \M. Let $u_j$ be the school which lowered its threshold in
this timestep.  Let $b_{u_j} = 1$ if and only if, for the unique
student $a_i$ such that $\score{u_j}{a_i} = t'_{u_j}$, it is true that
$u_j = \argmax_{\succ_{a_i}} \{u \mid \score{u}{a_i} \geq t_{u_j}\}$
($a_i$ prefers $u_j$ to all other schools for which her score
surpasses the threshold). Let $b_{u_{j'}} = -1$ if and only if $b_{u_j}
  = 1$ and also $u_{j'} = \secondargmax_{\succ_{a_i}} \{u \mid
  \score{u}{a_i} \geq t_{u_j}\}$ ($u_j$ is $a_i$'s favorite available
  school and $u_{j'}$ is her second favorite). For all other $j''$,
  let $b_{u_{j''}} = 0$.

  Then, there are at most $2m$ nonzero bits sent to \M about a
  particular student $a_i$, and at most $2$ nonzero bits sent by a
  particular $a_i$ to any school $u_j$. These bits are the only
  interface \M has with private data. Furthermore, \M and \privateda{\eps}{\delta} 
  have the same distribution over output data. So it suffices to show
  that \M is $(\eps, \delta)$-differentially private.

  Let $f : \{J\}^m \times [n] \to \{J\}^m$ be the function that, as a
  function of the previous thresholds and counter values, outputs the
  new set of thresholds at each time $t$.  Then, the thresholds
  published by \M are a composition of $f$, $m$ instantiations of
  \counter$(\eps', nmJ)$, and previously computed data. Thus, it
  suffices to show the composition of the $m$ counters satisfy
  $(\epsilon, \delta)$-differential privacy. By construction, each
  school $u_j$ receives at most $2$ nonzero bits from a given student,
  and no student's data creates more than $2m$ nonzero bits in all
  streams together. By Theorem~\ref{thm:composition} and Lemma~\ref{lem:l-comp},
  the composition of $m$ $\counter(\eps', n mJ)$ satisfy $(\eps,
  \delta)$-differential privacy when no stream has more than $2$ bits
  affected by a single agent's data and no student has more than $2m$
  total nonzero bits in any stream. Thus, \M (and also \privateda{\eps}{\delta})
  satisfies $(\eps, \delta)$ differential privacy.
\end{proof}

\begin{proof}[Lemma~\ref{lem:school-opt}]
  We prove that the output thresholds $\{t_{u_j}\}$ induce an
  $\alpha$-approximately stable, school-dominant matching $\mu$.
  Recall that $\mu$ is defined to be the matching where each student
  chooses her favorite school whose threshold she passes.

  We claim that there can be no blocking pairs with filled seats in
  $\mu$. Suppose some student $a_i$ wishes to attend $u_j$. Then, it
  is either the case that $\score{u_j}{a_i} \geq t_{u_j}$ or
  $\score{u_j}{a_i} < t_{u_j}$. In the first case, $a_i$ cannot block
  with $u_j$: she could have gone to $u_j$ and chose a school she
  preferred to $u_j$. In the second case, consider some $a_{i'}$ such
  that $u_j = \mu(a_{i'})$; this implies $\score{u_j}{a_{i'}} \geq
  t_{u_j}$. Thus, $\score{u_j}{a_{i'}} > \score{u_j}{a_i}$, so $a_{i'}
  \succ_{u_j} a_i$, and $a_i$ doesn't block with $a_{i'}$, so there
  are no blocking pairs with filled seats.

  By Theorem~\ref{thm:counters} and union bound, we know that the error of
  all $m$ counters over all time steps is bounded by $E$ except with
  probability $\beta$, where
  \[
  E =  \frac{128\sqrt{m\ln(1/\delta)}}{\eps} \ln\left(\frac{2m}{\beta}
        \right)\left( \sqrt{\log(nmJ)} \right)^5.
  \]

  So, we condition on the event that all schools' counters are
  accurate within $E$ throughout the run of \privateda{\eps}{\delta}
  for the remainder of our argument.

  We first claim that no school is over-enrolled in $\mu$. Consider the
  last time $u_j$ lowered its threshold to $t_{u_j}$. Let $n_{u_j}$
  denote the number of students tentatively matched to $u_j$ just
  before this final lowering of $t_{u_j}$.  By definition, $u_j$ only
  lowers its threshold when $\counter(u_j) < \hat{C}_{u_j} = C_{u_j} -
  E$, so

\[C_{u_j} - E = \hat{C}_{u_j} > \counter(u_j) \geq n_{u_j} - E \geq |\mu(u_j)| - E  - 1\]

where the first equality is by definition, the first inequality comes
from the fact that $u_j$ lowered its threshold, the third from the
accuracy we've conditioned on from the counters, and the final from
the fact that $u_j$ never again lowers its threshold. Thus, $C_{u_j}
\geq |\mu(u_j)|$, and $u_j$ is not over-enrolled.

Now, we show no school is under-enrolled by more than $2E$, unless
$t_{u_j} = 0$. When the algorithm terminates, each school $u_j$ either
has a threshold $t_{u_j} = 0$ or

\[|\mu(u_j)| + E \geq \counter(u_j) \geq \hat{C}_{u_j} = C_{u_j} - E\]

where the first equality comes from the conditional bound on the error
of the counters, the second from the fact that the algorithm
terminated, and the final one from the definition of
$\hat{C}_{u_j}$. Thus, $|\mu(u_j)| \geq C_{u_j} - 2E$ whenever
$t_{u_j} > 0$, so no school is under-enrolled by more than an
$\alpha$-fraction of its seats so long as

  \[
  C_{u_j} \geq 2E/\alpha.
  \]

  Finally, we show school dominance. We will now show that $\mu$, the
  matching corresponding to the thresholds output by
  $\privateda{\eps}{\delta}$, is also achieved by running \schoolda on
  the same instance, and halting early.

  No school is over-enrolled, by our argument above, \emph{at any point
    during the run of the algorithm}. So, each proposal made by $u_j$
  would be a valid proposal to make in \schoolda with full capacity.
  Thus, the algorithm terminates with each school having made (weakly)
  fewer proposals than it would have in \schoolda. Since each school
  makes its proposals in the same order (according to $\succ_{u_j}$),
  this implies that $\mu$ is a matching that corresponds to some
  intermediate point in \schoolda using with the same ordering of
  proposals. Thus, by Lemma~\ref{lem:prefix}, $\mu$ is
  school-dominant.  This argument is entirely parametric in $E$, the
  bounded error term from our counters, so the second part of the
  claim follows directly.
\end{proof}

Now, we proceed to prove the formal version of
Theorem~\ref{thm:k-list},  which is a
basic extension of Theorem~\ref{thm:formal} using the fact that
the sensitivity of School-Propose is reduced when students have fewer
schools in which they are interested.

\begin{theorem}\label{thm:kformal}
  Suppose each student has a preference list of length at most $k$.
  Then, \privateda{\frac{\epsilon \sqrt{m}}{2\sqrt{k}}}{\delta} is $(\epsilon,\delta)$-private and thus
  $(\epsilon+\delta)$-approximately truthful. With probability at least
  $1-\beta$, it outputs a set of $\alpha$-approximately stable
  admission thresholds that induces a school-dominant matching, as
  long as the capacity at each school
\[
C_u = \Omega\left(\frac{\sqrt{k}}{\alpha\epsilon}\polylog\left(
    n, m, \frac{1}{\delta}, \frac{1}{\beta}\right)\right)
\]
\end{theorem}
\begin{remark}\label{remark:limit}
  We assume that the algorithm ignores students after they have
  accepted $k$ or more schools' proposals.
\end{remark}

\begin{proof}[Theorem~\ref{thm:k-list}]
  We prove $(\epsilon, \delta)$-privacy, which again reduces to
  proving $(\epsilon, \delta)$-differential privacy of the set of $m$
  counters. By a simple calculation, \privateda{\frac{\epsilon
    \sqrt{m}}{2\sqrt{k}}}{\delta} uses

\[\epsilon' =
\frac{\epsilon}{4\sqrt{k\ln{\frac{1}{\delta}}}}\]

as the privacy parameter for the $m$ counters it uses.
Theorem~\ref{thm:composition} states that a collection of $m$
$\counter(\epsilon', nT)$ with total sensitivity $\Delta$ (and
individual sensitivity $c$) satisfy $(\epsilon, \delta)$-differential
privacy so long as

  \[\epsilon' \leq \frac{\epsilon}{2c \sqrt{\Delta\ln{\frac{1}{\delta}}}}.\]

  Remark~\ref{remark:limit} limits the total amount of sensitivity
  School-Propose will have; a student will be able to affect at most
  $2k$ bits in the input stream, so $\Delta \leq 2k$, and at most $2$
  per school, so $c\leq 2$. Thus, it suffices to use privacy parameter

\[\epsilon' \leq \frac{\epsilon}{4\sqrt{k\ln{\frac{1}{\delta}}}},\]

so our algorithm is $(\epsilon, \delta)$-differentially
private. Furthermore, Theorem~\ref{thm:composition} and a union bound imply
the maximum error any one of the counters will have at any time during
the execution of the algorithm is

\[E \leq \frac{128\sqrt{k\ln(1/\delta)}}{\eps} \ln\left(\frac{2m}{\beta}
        \right)\left( \sqrt{\log(nmJ)} \right)^5. \]

with probability $1-\beta$. Thus, by Lemma~\ref{lem:school-opt}, we
get the desired guarantee for $\alpha$-approximate stability and
school-dominance.
\end{proof}

\section{Conclusions}
In this paper we applied differential privacy as a tool to design a
many-to-one stable matching algorithm with strong incentive guarantees
for the student side of the market. To the best of our knowledge, our
work is the first work to show positive truthfulness results for the
non-optimal side of the market, under \emph{worst-case} preferences,
for \emph{all} participants on the non-optimal sie of the market.

Additionally, although we have not focused on this, our algorithm also
provides strong \emph{privacy} guarantees to the students. Each
student, upon learning the school thresholds (and hence the school
that she herself is matched to) can learn almost nothing about either
the preferences or scores of the other students (i.e. almost nothing
about the preferences that the other students hold over schools, or
the preferences that schools hold over the other students). Here
``almost nothing'' is the formal guarantee of differential privacy,
which in particular implies that for every student $a$, no matter what
her prior belief over the private data of some other student $a'$ is,
her posterior belief over $a'$s data would be almost the same in the
two worlds in which $a'$ participates in the mechanism, and in which
she does not. These guarantees might themselves be valuable in
settings in which the matching being computed is sensitive --
e.g. when computing a matching between patients and drug trials, for
example.

\bibliographystyle{plainnat}

\bibliography{sources}

\appendix
\section{Privacy Analysis for Counters}

\citet{chan-counter} show that $\counter(\epsilon, T)$ is
$\epsilon$-differentially private with respect to single changes in
the input stream, when the stream is generated non-adaptively. For our
application, we require privacy to hold for a large number of streams
whose joint-sensitivity can nevertheless be bounded, and whose entries
can be chosen adaptively. To show that \counter is also private in
this setting (when $\epsilon$ is set appropriately), we first present
a slightly more refined composition theorem.

\subsection{Composition}
\label{sec:compose}
An important property of differential privacy is that it degrades
gracefully when private mechanisms are composed together, even
adaptively. We recall the notion of an \emph{adaptive composition
experiment} due to \citet{boostingDP}:

\begin{itemize}
  \item Fix a bit $b \in \{0, 1\}$ and a class of mechanisms $\cM$.
  \item For $t = 1 \dots T$:
    \begin{itemize}
      \item The adversary selects two databases $D^{t, 0}, D^{t, 1}$ and a
        mechanism $\cM_t \in \cM$.
      \item The adversary receives $y_t = \cM_t(D^{t, b})$
    \end{itemize}
\end{itemize}
%
The ``output'' of an adaptive composition experiment is the view of the
adversary over the course of the experiment. The experiment is said to be
$\epsilon$-differentially private if
\[
  \max_{S \subseteq \cR}\frac{\Pr[V^0 \in S]}{\Pr[V^1 \in S]} \leq
  \exp(\epsilon),
\]
and $(\eps, \delta)$-differentially private if
\[
\max_{S\subset \cR, \Pr[V^0\in S]\geq \delta} \frac{\Pr[V^0\in S] -
  \delta}{\Pr[V^1\in S]} \leq \exp(\eps),
\]
where $V^0$ is the view of the adversary with $b = 0$, $V^1$ is the
view of the adversary with $b = 1$, and $\cR$ is the range of outputs.

Any algorithm that can be described as an instance of this adaptive
composition experiment (for an appropriately defined adversary) is
said to be an instance of the class of mechanisms $\cM$ under
\emph{adaptive $T$-fold composition}.

A very useful tool to analyze private algorithms is the following
theorem that allows us to analyze the ``composition'' of private
algorithms.

\begin{theorem}[\citet{boostingDP}]
\label{thm:composition}
Let $\mathcal{A}\colon \mathcal{U} \rightarrow \mathcal{R}^T$ be a
$T$-fold adaptive composition\footnote{
See Section~\ref{sec:compose} and \citep{boostingDP} for further discussion.
} of $(\eps, \delta)$-differentially
private algorithms. Then $\mathcal{A}$ satisfies $(\eps', T\delta +
\delta')$-differential privacy for
\[
\eps' = \eps\sqrt{2T\ln(1/\delta')} + T\eps(e^\eps - 1).
\]
In particular, for any $\eps\leq 1$, if $\mathcal{A}$ is a $T$-fold
adaptive composition of $(\eps/\sqrt{8T\ln(1/\delta)},
0)$-differentially private mechanisms, then $\mathcal{A}$ satisfies
$(\eps, \delta)$-differential privacy.
\end{theorem}


For a more refined analysis in our setting, we now state a straightforward
consequence of a composition theorem of \citet{boostingDP}.

\begin{lemma}[\citet{boostingDP}]
  \label{lem:l-comp}
  Let $\Delta \geq 0$. Under adaptive composition, the class of
  $\frac{\epsilon}{\Delta}$-private mechanisms satisfies
  $\epsilon$-differential privacy and the class of
  $\frac{\eps}{2c\sqrt{2\Delta\ln(1/\delta)}}$-private mechanisms
  satisfies $(\eps, \delta)$-differential privacy, if the adversary
  always selects databases satisfying
  \[
   \mbox{for all } t\; \left|D^{t,0} - D^{t, 1} \right|\leq c, 
   \]
   \[
   \mbox{and also }\;     \sum_{t = 1}^T \left|D^{t, 0} - D^{t, 1}\right| \leq \Delta.
 \]
\end{lemma}
In other words, the privacy parameter of each mechanism should be calibrated for
the total distance between the databases, over the whole composition.
This is useful for analyzing the privacy of the counters in our
algorithm, which collectively have bounded sensitivity.

\subsection{Details for Counters}
\label{counter-details}

We reproduce Binary mechanism here in order to refer to its internal
workings in our privacy proof.

First, it is worth explaining the intuition of the \counter. Given a
bit stream $\sigma \colon [T] \rightarrow \{-1, 0, 1\}$, the algorithm
releases the counts $\sum_{i=1}^t \sigma(i)$ for each $t$ by
maintaining a set of partial sums $\sum[i, j] \coloneqq \sum_{t=i}^j
\sigma(t)$. More precisely, each partial sum has the form $\sum[2^i
+ 1, 2^i + 2^{i - 1}]$, corresponding to powers of $2$.

In this way, we can calculate the count $\sum_{i=1}^t \sigma(i)$ by
summing at most $\log{t}$ partial sums: let $i_1 < i_2 \ldots < i_m$
be the indices of non-zero bits in the binary representation of $t$,
so that
\begin{align*}
  \label{eq:binary}
  \sum_{i=1}^t  \sigma(i)&= \sum[1, 2^{i_m}] + \sum[2^{i_m}+1, 2^{i_m}
  + 2^{i_{m-1}}]\\& + \ldots + \sum[t-2^{i_1} + 1, t].
\end{align*}
Therefore, we can view the algorithm as releasing partial sums of
different ranges at each time step $t$ and computing the counts is
simply a post-processing of the partial sums. The core algorithm is
presented in Algorithm~\ref{alg:binary-mechanism}.


\begin{algorithm}[h!]
     \caption{$\counter(\eps, T)$}
     \begin{algorithmic}\label{alg:binary-mechanism}
       \STATE{\textbf{Input:} A stream $\sigma\in \{-1,1\}^T$}
       \STATE{\textbf{Output: } $B(t)$, estimate of $\sum_{i=1}^t
         \sigma(i)$ at each $t\in [T]$}
       \FORALL{$t\in [T]$}
       \STATE Express $\displaystyle t = \sum_{j=0}^{\log{t}} 2^j\bin_j(t)$.
       \STATE Let $i \leftarrow \min_j\{\bin_j(t) \neq 0\}$
       \STATE $a_i\leftarrow \sum_{j < i} a_j + \sigma(t) $,
               $(a_i  = \sum[t-2^i + 1, t])$
               \FOR{$0\leq j \leq i - 1$}
               \STATE Let $a_j \leftarrow 0$ and $\hat{a_j} \leftarrow   0$
     \ENDFOR
     \STATE Let $\hat{a_j} = a_j + \Lap(\log(T) /\eps)$
     \STATE Let $\displaystyle B(t) = \sum_{i: \bin_i(t)\neq 0} \hat{a_i}$
      \ENDFOR

     \end{algorithmic}
\end{algorithm}

\subsection{Counter Privacy under Adaptive Composition}

\countercomp*

\begin{proof}
  The composition of $m$ counters is essentially releasing a
  collection of noisy partial sums adaptively. We need to first frame
  this setting as an advanced composition experiment defined in
  Appendix~\ref{sec:compose}. First, we treat each segment $\sigma[a,b]$ in a
  stream as a database. For each such database, we are releasing the
  sum by adding noise sampled from the Laplace distribution:
  \[
  \Lap\left(\frac{2c\sqrt{2kc\ln(1/\delta)}\log(T)}{\eps}\right),
  \]
  which is $\frac{\eps}{2c\sqrt{2kc \ln(1/\delta)\log(T)}}$-private
  mechanism (w.r.t.\ a single bit change). We know that changing an
  agent's data changes at most $c$ bits in each stream, and affects at
  $k$ streams, and also each bit change can result in $\log(T)$ bits
  changes across different stream-segment databases. Therefore, we can
  bound the total distance between all pairs stream-segment databases
  by
  \[
  \Delta \leq k c \log(T).
  \]
  By Lemma~\ref{lem:l-comp}, we know that the composition of all $m$
  counters under our condition satisfies $(\eps, \delta)$-differential
  privacy. 
\end{proof}

\section{Proofs of Matching Lemmas}\label{sec:matching}

We state one more lemma which we will use in the proofs of
Lemmas~\ref{lem:prefix}.

\begin{lemma}\label{lem:equiv}
  Consider a set $P_u$ of proposals made by each school $u$ according
  to some prefix of school-proposing $DA$. Let $P\subseteq A \times U$
  be the set of proposals made by all schools. Then, the matching
  $\mu$ which results from $P$ is unique (and independent of the order
  in which proposals are made), assuming students are truthful.
\end{lemma}
\begin{proof}
  Each student ultimately accepts her most preferred proposal among the set of proposals she has received, independent of their ordering. (i.e. admissions thresholds only descend, and she picks her most preferred school amongst those schools with thresholds below her scores). Thus, each school $u$
  will be matched to the subset of $P_u$ which finds $u$ to be their
  favorite offer, independent of the order in which proposals were made.
\end{proof}

\prefix

\begin{proof}
  The school-proposing deferred acceptance algorithm is somewhat
  underspecified. In particular, if multiple schools have space
  remaining, the \emph{order} in which those schools make proposals
  isn't predetermined. But, by Lemma~\ref{lem:equiv} shows that
  reordering of the same proposals from the schools will arrive at the
  same matching. Thus, it suffices to show, for a fixed ordering of
  the entire set of proposals made by DA, that each intermediate
  matching is school-dominant.

  Let $t$ denote the time at which we wish to halt a run of DA. Let
  $P_{u,t}$ denote the set of proposals which school $u$ has made
  according some fixed ordering up to time $t$, and $P_{u}$ denote set
  of proposals made by school $u$ according to the entire run of
  DA. Let $\mu_t$ denote the ``current'' matching according to the
  first run of DA stopped at time $t$ and $\mu$ denote the final
  outcome of DA.

  Consider any school $A$. Notice that, since $|P_{u}| \geq |P_{u,
    t}|$, by the definition of DA,

  \begin{align}P_{u,t} \subseteq
    P_{u}\label{eqn:subset}\end{align}

  since, for a given school, the proposing order is just working down
  their preference list.

  Now consider a particular school $u$. We must show that for each $a
  \in \mu_t(u) \setminus \mu(u), a' \in \mu(u) \setminus \mu_{t}(u)$,
  $a \succ_u a'$. If $a$ was proposed to by $u$ in $P_t$ and rejects
  $u$, then $u$ will be rejected by $a$ when she receives a superset
  $P$ of proposals. Thus, the only students $u$ has according to $\mu$
  but not $\mu_t$ are students $z\in P_{u} \setminus P_{u,t}$
  (students who are proposed to after time $t$). But, by the
  definition of $DA$, if $u$ proposes to two students $a$ and $a'$,
  and proposes to $a$ before $a'$, $a \succ_u a'$, as desired.
\end{proof}

\section{Private Matching Algorithms Must Allow Empty Seats}
In this paper, we gave an algorithm with strong worst-case incentive properties in large markets, without needing to make distributional assumptions about the agents preferences, or requiring any other ``large market'' condition other than that the capacities of the schools be sufficiently large. However, in exchange, we had to relax our notion of stability to an approximate notion which allows a small number of empty seats per school. We here give an example demonstrating why this relaxation is necessary for any differentially private matching algorithm. An algorithm that must return an -exactly- stable matching must have extremely high sensitivity to the change in preferences of any single agent, if preferences are allowed to be worst case.

\begin{example}\label{example:empty}
Suppose there are $n$ students and $2$ schools, $H$ and $Y$. Suppose, for students $1
\leq a \leq \frac{n}{2}$, $H \succ_{a} Y$, and for $\frac{n}{2} < a
\leq n$, $Y \succ_{a} H$. Each school has capacity for exactly half
of the students: $C_H = C_Y = \frac{n}{2}$. Suppose $Y$ has preference
ordering $\succ_Y$, $s_1 \succ_Y s_2 \succ_Y \ldots \succ_y s_n$; $H$
has preference ordering $s_{\frac{n}{2}+1} \succ_H s_{\frac{n}{2}+2}
\succ_H \ldots s_n \succ_H s_1 \succ_H \ldots \succ_H
s_{\frac{n}{2}}$. The school-optimal matching matches students $s_1,
\ldots, s_{\frac{n}{2}}$ to $Y$ and $s_{\frac{n}{2} +1}, \ldots, s_n$
to $H$. Now consider the market with any single student removed. The school-optimal
stable matching changes entirely (i.e. every single student is matched to a different school). For example, if $s_1$ is
removed, $Y$ will admit $s_{\frac{n}{2}+1}$ (who will accept), $H$
will admit $s_{2}$ (who will accept), $Y$ will admit $s_{\frac{n}{2} +
  2}$ and so on. In the end, each student will get her favorite
school, and the schools will swap students. The same effect is achieved by having a single student \emph{change} her preferences, by reporting that she prefers to be unmatched than to be matched to her second choice school. This example shows that the exact
school-optimal matching is \emph{highly} sensitive to the addition,
removal, or alteration of preferences of a single student and hence impossible to achieve under differential privacy. Our algorithms blunt this kind of sensitivity via the use of a small budget of seats that we may leave empty. 
\end{example}

\end{document}